\documentclass{ifacconf}

\usepackage{graphicx}      
\usepackage{natbib}        

\usepackage{amsmath}
\usepackage{amssymb}
\usepackage{booktabs}
\usepackage{tikz}
\usetikzlibrary{arrows.meta}
\usetikzlibrary{decorations.pathreplacing}
\usetikzlibrary{calc}
\newcommand{\penetration}{\ensuremath{y}}
\newcommand{\sigProb}{\ensuremath{\beta}}
\newcommand{\accidentCost}{\ensuremath{r}}
\newcommand{\accidentFunc}{\ensuremath{p}}
\newcommand{\detectionFunc}{\ensuremath{q}}

\newcommand{\equiProfile}{\ensuremath{x}}
\newcommand{\randVarMass}{\ensuremath{D}}
\newcommand{\du}{\ensuremath{d_{\mathrm{u}}}}
\newcommand{\ds}{\ensuremath{d_{\mathrm{s}}}}
\newcommand{\totalGroupM}{\ensuremath{\bar{x}}}

\newcommand{\groupN}{\ensuremath{\mathrm{n}}}
\newcommand{\groupVU}{\ensuremath{\mathrm{vu}}}
\newcommand{\groupVS}{\ensuremath{\mathrm{vs}}}
\newcommand{\groupSet}{\ensuremath{\mathcal{G}}}
\newcommand{\groupEle}{\ensuremath{g}}

\newcommand{\actionRec}{\ensuremath{\mathrm{R}}}
\newcommand{\actionCare}{\ensuremath{\mathrm{C}}}
\newcommand{\actionEle}{\ensuremath{a}}

\newcommand{\expC}{\ensuremath{J}}

\newcommand{\eventAcci}{\ensuremath{A}}
\newcommand{\eventDisp}{\ensuremath{S}}
\newcommand{\eventDispNeg}{\ensuremath{\neg\eventDisp}}
\newcommand{\probability}[1]{\ensuremath{\Pr(#1)}}
\newcommand{\condProb}[2]{\ensuremath{\Pr(#1 \mid #2)}}

\newcommand{\TOne}{\ensuremath{\expP_\groupN}}
\newcommand{\TTwo}{\ensuremath{\expP_\groupVU}}
\newcommand{\TOneVal}{\ensuremath{\frac{1}{1+\accidentCost}}}
\newcommand{\TTwoVal}{\ensuremath{\frac{1}{1+\accidentCost(1-\sigProb\detectionFunc(\penetration))}}}

\newcommand{\expP}{\ensuremath{P}}
\newcommand{\expPNE}{\ensuremath{\expP^{\mathrm{ne}}}}
\newcommand{\candProbOneVal}{\ensuremath{\frac{\accidentFunc(1)}{1+\sigProb\detectionFunc(\penetration)\bigl(\accidentFunc(1)-\accidentFunc(1-\penetration)\bigr)}}}
\newcommand{\candProbTwoVal}{\ensuremath{\frac{\accidentFunc(\equiProfile_\groupN+\penetration)}{1+\sigProb\detectionFunc(\penetration)\bigl(\accidentFunc(\equiProfile_\groupN+\penetration)-\accidentFunc(\equiProfile_\groupN)\bigr)}}}
\newcommand{\candProbThreeVal}{\ensuremath{\frac{\accidentFunc(\penetration)}{1+\sigProb\detectionFunc(\penetration)\bigl(\accidentFunc(\penetration)-\accidentFunc(0)\bigr)}}}
\newcommand{\candProbFourVal}{\ensuremath{\frac{\accidentFunc(\equiProfile_\groupVU)}{1+\sigProb\detectionFunc(\penetration)\bigl(\accidentFunc(\equiProfile_\groupVU)-\accidentFunc(0)\bigr)}}}
\newcommand{\candProbFiveVal}{\ensuremath{\accidentFunc(0)}}

\newcommand{\consistencyEq}{\ensuremath{\frac{\accidentFunc(\du)}{1+\sigProb\detectionFunc(\penetration)\bigl(\accidentFunc(\du)-\accidentFunc(\ds)\bigr)}}}

\DeclareMathOperator{\NE}{NE}
\newcommand{\regionSet}{\ensuremath{\mathcal{R}}}
\newcommand{\Tstate}{\ensuremath{B}}
\newcommand{\candSet}{\ensuremath{X}}
\newcommand{\paraCond}{\ensuremath{C}}

\newcommand{\effSig}{\ensuremath{k}}
\newcommand{\utilityU}{\ensuremath{u}}
\newcommand{\utilityV}{\ensuremath{v}}

\newcommand{\profileInput}{\ensuremath{(\equiProfile; \sigProb, \penetration)}}
\newcommand{\paraInput}{\ensuremath{(\sigProb, \penetration)}}
\newcommand{\optimalSigP}{\ensuremath{\sigProb^*}}

\begin{document}
\begin{frontmatter}

\title{Adverse Effects of V2V Adoption on Road Safety\thanksref{ifacnotice}} 


\thanks[ifacnotice]{
	\textcopyright{} 2026 the authors.
	This work has been accepted to IFAC for publication under a Creative Commons Licence CC-BY-NC-ND. The accepted CPHS 2026 conference paper is an abbreviated version of this extended manuscript, which contains complete proofs.
}

\author[UCCS]{Zhenqi Liu} 
\author[UCCS,Polito]{Philip N. Brown} 
\author[UCCS]{Keith Paarporn}

\address[UCCS]{Computer Science Department, 
	University of Colorado Colorado Springs, 
	Colorado Springs, CO 80918 USA \\
	(e-mail: {\tt\{zliu3, pbrown2, kpaarpor\}@uccs.edu})}
\address[Polito]{DISMA, Politecnico di Torino, Italy}

\begin{abstract}                
Vehicle-to-vehicle (V2V) communication is expected to improve road safety and reduce congestion. However, prior work shows that V2V information sharing under partial adoption may increase congestion and decrease safety. We study whether increasing V2V adoption itself affects road safety. We propose a corrected version of an existing model and analyze its behavior under varying adoption levels. We show that, in some cases, increased V2V adoption can increase accident probability. However, under an optimal signaling policy, the system can ensure that accident probability is non-increasing in the adoption level.
\end{abstract}


\end{frontmatter}

\section{Introduction}
\label{sec:introduction}
Vehicle-to-vehicle (V2V) communication is an emerging smart transportation technology that allows equipped vehicles to exchange information about road hazards. Such systems have the potential to improve road safety by warning drivers of road hazards or dangerous conditions before they encounter them. However, transportation systems are shaped not only by technical performance, but also by the behavioral responses of strategic drivers. When drivers change their behavior in response to information, the resulting aggregate behavior can feed back into the safety or congestion state of the road system. Thus, understanding the effect of V2V communication requires an equilibrium perspective.

Equilibrium models of transportation systems go back at least to Wardrop's classical principles, which distinguish individually stable route choices from system-optimal traffic assignments \citep{wardrop1952road}. A large literature on selfish routing and traffic equilibrium shows that individually optimal driver behavior need not produce system-optimal outcomes \citep{correa2004selfish}, and that transportation equilibria can exhibit counterintuitive phenomena such as Braess-type paradoxes \citep{dafermos1984some}. Information can introduce an additional source of counterintuitive behavior. Driver information systems may improve route choice and network performance, but they may also create concentration, overreaction, or other adverse effects \citep{ben1991dynamic,balakrishna2013information}. In particular, heterogeneous access to information can change equilibrium outcomes in nontrivial ways \citep{liu2016effects}, and providing additional information to a subset of drivers can even make informed drivers worse off, a phenomenon known as informational Braess' paradox \citep{acemoglu2018informational}.

These observations connect naturally to the literature on information design. Bayesian persuasion studies how a sender can influence a receiver's action by choosing what information to reveal \citep{kamenica2011bayesian}, and information design provides a broader framework for studying how information structures affect strategic behavior \citep{bergemann2019information}. Recent work has applied these ideas to transportation networks, studying how public signals, private recommendations, or partial participation in information systems can be used to regulate equilibrium traffic flows \citep{wu2019information,massicot2021competitive,zhu2022information}. These works show that full information disclosure is not always optimal and that the design of information provision can substantially affect equilibrium performance.

Our work is most closely related to the V2V hazard-warning model of \citet{gould2022partial}, which studies partial adoption of V2V communication in a nonatomic driver population. In that model, only a fraction of drivers are V2V-equipped, and equipped drivers may condition their behavior on whether a warning signal is displayed. They show that sharing hazard information only with V2V-equipped drivers can paradoxically increase the equilibrium frequency of accidents relative to no information sharing. An extended journal version further develops this information-design perspective for V2V communication, including false positives and social-cost considerations \citep{gould2023information}. These results are part of a broader theme in which emerging transportation technologies can have unintended equilibrium effects \citep{mehr2019will}.

In this paper, we revisit the partial-adoption V2V signaling model under a corrected accident-probability consistency equation. Motivated by the adverse effects identified in \citet{gould2022partial}, we ask whether the original consistency model should be revised, what happens as V2V adoption increases under the corrected model, and whether the resulting adverse effects can be mitigated. Our main contributions are as follows:
\begin{itemize}
	\item We identify a consistency issue in the original accident-probability model: unsignaled and signaled V2V behaviors are contingent on different information states and should not be averaged as simultaneous reckless masses. We formulate the corrected consistency equation.
	\item Under the corrected model, we derive an explicit equilibrium-region characterization and use it to analyze the effect of V2V adoption. This characterization shows that, under fixed signaling probabilities, increasing V2V adoption can increase the equilibrium accident probability, and that any such local perverse effect can occur only in two specific equilibrium regions.
	\item We show that the perverse effect can be eliminated by optimizing the signaling probability. Specifically, if the signaling probability is chosen optimally at each adoption level, then the optimized equilibrium accident probability is weakly decreasing in V2V adoption.
\end{itemize}

The rest of the paper is organized as follows. Section~\ref{sec:model} presents the corrected V2V signaling model. Section~\ref{sec:main-results} states the main results. Section~\ref{sec:proofs} provides the equilibrium characterization and proofs. Section~\ref{sec:discussion} discusses implications of the results, and Section~\ref{sec:conclusion} concludes.

\section{Model}
\label{sec:model}
\subsection{Model Setup}
The population of drivers is modeled as a continuum with total mass one. A fraction \(\penetration\in[0,1]\) of drivers are V2V-equipped, which we call V2V drivers, while the remaining fraction \(1-\penetration\) are not, which we call non-V2V drivers. We refer to \(\penetration\) as the V2V \emph{adoption level}. A driver can choose either to drive carefully or recklessly. We let \(\actionEle\in\{\actionCare,\actionRec\}\) denote a driver's action, where \(\actionCare\) denotes careful driving and \(\actionRec\) denotes reckless driving.

Let \(\eventAcci\) denote the event that an accident occurs. Conditional on an accident occurring, the V2V system may detect it. We write the \emph{detection probability} as \(\detectionFunc(\penetration)\), reflecting its dependence on the V2V adoption level. Throughout this manuscript, we assume that \(\detectionFunc:[0,1]\to[0,1]\) is continuous and weakly increasing, so a higher V2V adoption level makes accident detection weakly more likely. We do not impose strict monotonicity unless explicitly stated.

In the V2V setting, once an accident is detected, the system designer can choose whether to display the detected accident information to V2V-equipped drivers. We denote by \(\sigProb\in[0,1]\) the \emph{signaling probability} that this information is displayed to V2V-equipped drivers.

Let \(\eventDisp\) denote the event that a warning signal is shown to a V2V driver. This event happens only when an accident occurs, the V2V system detects it, and the system displays the detected accident information with probability \(\sigProb\). Therefore,
\[
\probability{\eventDisp}
=
\probability{\eventAcci}\detectionFunc(\penetration)\sigProb.
\]
In our model, the accident probability \(\probability{\eventAcci}\) is endogenously determined. Driver behavior influences the accident probability, and the accident probability, together with the primitive game parameters, in turn influences driver behavior. Careful drivers do not negatively affect road safety, while the total mass of reckless drivers determines how likely an accident is to occur. Thus, let \(\equiProfile\) denote the \emph{behavior profile} encoding the mass of reckless drivers, and let \(\expP\profileInput\) denote the aggregate accident probability induced by \(\equiProfile\), \(\sigProb\), and \(\penetration\). Then \(\probability{\eventAcci}=\expP\profileInput\), and
\begin{equation}
	\label{eq:signal-display-probability}
	\probability{\eventDisp}
	=
	\expP\profileInput\detectionFunc(\penetration)\sigProb.
\end{equation}
The population is naturally divided into non-V2V drivers and V2V-equipped drivers. Non-V2V drivers do not receive warning signals, while V2V-equipped drivers may condition their behavior on whether a warning signal is displayed. Thus, a V2V driver has two possible information states: unsignaled and signaled. We denote by \(\equiProfile_\groupN\in[0,1-\penetration]\) the reckless mass of non-V2V drivers, by \(\equiProfile_\groupVU\in[0,\penetration]\) the reckless V2V mass on the unsignaled event \(\eventDispNeg\), and by \(\equiProfile_\groupVS\in[0,\penetration]\) the reckless V2V mass on the signaled event \(\eventDisp\). The behavior profile is therefore written as
\[
\equiProfile = (\equiProfile_\groupN,(\equiProfile_\groupVU,\equiProfile_\groupVS)).
\]
Although one may informally refer to non-V2V, unsignaled V2V, and signaled V2V drivers as separate groups, this notation should be understood as a contingent behavior profile of the V2V population under the unsignaled and signaled events. In particular, \(\equiProfile_\groupVU\) and \(\equiProfile_\groupVS\) are not added simultaneously in a realized information state. Therefore, we define the unsignaled-state and signaled-state total reckless masses by
\[
\du(\equiProfile):=\equiProfile_\groupN+\equiProfile_\groupVU,
\qquad
\ds(\equiProfile):=\equiProfile_\groupN+\equiProfile_\groupVS.
\]
When the profile \(\equiProfile\) is clear from context, we simply write \(\du\) and \(\ds\).

We define the function \(\accidentFunc:[0,1]\to[0,1]\), which maps a realized total reckless driver mass to the corresponding accident probability. Throughout the manuscript, we assume that \(\accidentFunc\) is continuous and weakly increasing, so a larger reckless driver mass makes an accident weakly more likely. We do not impose strict monotonicity unless explicitly stated.

The consistency equation for \(\expP\profileInput\) is then
\begin{equation}
	\label{eq:consistency-equation-new}
	\expP\profileInput
	=
	\bigl(1-\probability{\eventDisp}\bigr)\accidentFunc(\du)
	+
	\probability{\eventDisp}\accidentFunc(\ds).
\end{equation}
This consistency equation can be solved explicitly for \(\expP\profileInput\). Substituting \eqref{eq:signal-display-probability} into \eqref{eq:consistency-equation-new} and solving for \(\expP\profileInput\) gives
\begin{equation}
	\label{eq:consistency-equation-explicit}
	\expP\profileInput = \consistencyEq.
\end{equation}
The expression in \eqref{eq:consistency-equation-explicit} removes the recursive appearance of \(\expP\profileInput\). In Lemma~\ref{lem:boundedness-of-consistency-equation} of Section~\ref{sec:proofs-region-characterization}, we verify that this expression is well-defined and yields a valid probability for the equilibrium characterization.

\subsection{Equilibrium Definition}
A driver who chooses \(\actionRec\) incurs an expected cost \(\accidentCost\) when an accident occurs and cost \(0\) otherwise. A driver who chooses \(\actionCare\) incurs regret cost \(1\) when no accident occurs and cost \(0\) otherwise. We normalize the regret cost to \(1\) and assume that the expected accident cost is strictly larger, so \(\accidentCost>1\). Let \(\groupEle\in\groupSet:=\{\groupN,\groupVU,\groupVS\}\) denote a driver group, and let \(\expC_\groupEle(\equiProfile;\actionEle)\) denote the expected cost of action \(\actionEle\) for group \(\groupEle\) under profile \(\equiProfile\). The expected costs are summarized in Table~\ref{tb:expected-cost}.
\begin{table}[h]
	\centering
	\caption{Expected cost by driver group and action.}
	\label{tb:expected-cost}
	\begin{tabular}{c|ccc}
		\toprule
		\text{Action} 
		& \(\groupN\)
		& \(\groupVU\)
		& \(\groupVS\) \\
		\midrule
		\(\actionCare\)
		& \(1-\probability{\eventAcci}\)
		& \(1-\condProb{\eventAcci}{\eventDispNeg}\)
		& \(1-\condProb{\eventAcci}{\eventDisp}\) \\[3pt]
		\(\actionRec\)
		& \(\accidentCost\probability{\eventAcci}\)
		& \(\accidentCost\condProb{\eventAcci}{\eventDispNeg}\)
		& \(\accidentCost\condProb{\eventAcci}{\eventDisp}\) \\
		\bottomrule
	\end{tabular}
\end{table}

We assume that warning signals have zero false-positive rate: whenever a warning signal is displayed, an accident has occurred. Hence
\[
\condProb{\eventAcci}{\eventDisp}=1.
\]
Under this assumption, the expected costs for signaled V2V drivers simplify to
\[
\expC_\groupVS(\equiProfile;\actionCare)
=
1-\condProb{\eventAcci}{\eventDisp}
=
0, \text{ and }
\]
\[
\expC_\groupVS(\equiProfile;\actionRec)
=
\accidentCost\condProb{\eventAcci}{\eventDisp}
=
\accidentCost.
\]
Therefore, since \(\accidentCost>1\), signaled V2V drivers strictly prefer careful driving. This preference follows from the certainty of the accident that generates the warning, whereas \(\accidentFunc(\ds)\) describes the risk of subsequent accidents arising from the resulting driver behavior and thus need not equal \(\condProb{\eventAcci}{\eventDisp}\).

Let \(\totalGroupM_\groupEle\) denote the total mass of group \(\groupEle\). Then
\[
\totalGroupM_\groupN=1-\penetration,
\qquad
\totalGroupM_\groupVU=\penetration,
\qquad
\totalGroupM_\groupVS=\penetration.
\]
For each \(\groupEle\in\groupSet\), the reckless mass satisfies \(0\le \equiProfile_\groupEle\le \totalGroupM_\groupEle\). We say \(\equiProfile\) is an equilibrium profile if each group assigns positive mass only to cost-minimizing actions. This can be written compactly as
\begin{equation}
	\label{eq:exp-cost-reckless-optimal}
	\equiProfile_\groupEle>0
	\implies
	\expC_\groupEle(\equiProfile;\actionRec) \leq
	\expC_\groupEle(\equiProfile;\actionCare),
\end{equation}
\begin{equation}
	\label{eq:exp-cost-careful-optimal}
	\equiProfile_\groupEle<\totalGroupM_\groupEle
	\implies
	\expC_\groupEle(\equiProfile;\actionCare) \leq
	\expC_\groupEle(\equiProfile;\actionRec),
\end{equation}
for every \(\groupEle\in\groupSet\). In words, if a positive mass of group \(\groupEle\) drives recklessly, then reckless driving must be weakly cost-minimizing; if a positive mass of group \(\groupEle\) drives carefully, then careful driving must be weakly cost-minimizing.

\subsection{Corrected Accident Probability Model}
We retain the behavioral and cost structure of the original model \citep{gould2022partial}, but replace its consistency equation with \eqref{eq:consistency-equation-new}. The original consistency equation is
\[
\expP\profileInput
=
\accidentFunc\Bigl(
\equiProfile_\groupN
+
\bigl(1-\probability{\eventDisp}\bigr)\equiProfile_\groupVU
+
\probability{\eventDisp}\equiProfile_\groupVS
\Bigr).
\]
To explain the correction, define the random variable \(\randVarMass\) as the  realized total reckless mass. For a given profile \(\equiProfile\), let
\[
\randVarMass
=
\begin{cases}
	\du(\equiProfile), & \text{on } \eventDispNeg,\\
	\ds(\equiProfile), & \text{on } \eventDisp.
\end{cases}
\]
The original formulation effectively evaluates accident risk at an averaged reckless mass. In contrast, the corrected consistency equation evaluates the accident function separately in the unsignaled and signaled states and then averages the resulting accident probabilities. This distinction matters because \(\accidentFunc\) maps the realized reckless mass in a given information state to the accident probability in that state. The accident probability in one realized information state should not be evaluated using reckless masses from other states.

In simple terms, the original consistency equation evaluates accident probability as \(\accidentFunc\bigl(\mathbb E[\randVarMass]\bigr)\), whereas the corrected consistency equation evaluates accident probability as \(\mathbb E\bigl[\accidentFunc(\randVarMass)\bigr]\). In general, 
\begin{equation}
	\label{eq:old-new-consistency-diff}
	\mathbb E\bigl[\accidentFunc(\randVarMass)\bigr]
	\neq
	\accidentFunc\bigl(\mathbb E[\randVarMass]\bigr).
\end{equation}
When \(\accidentFunc\) is affine, the original and corrected consistency equations coincide by the linearity of expectation. For nonlinear \(\accidentFunc\), the two formulations can yield different accident probabilities and equilibrium-region conditions. Nevertheless, several qualitative conclusions of \citet{gould2022partial} remain valid under the corrected formulation, including the equilibrium behavioral structure and the endpoint optimality of \(\sigProb\) for minimizing equilibrium accident probability, as established in Section~\ref{sec:proofs}.

\subsection{Notation}
Although \(\expP\) depends on all primitive parameters, so that
\[
\probability{\eventAcci}
=
\expP(\equiProfile;\sigProb,\penetration,\accidentCost,\accidentFunc,\detectionFunc),
\]
we primarily study how the induced accident probability varies with the signaling probability \(\sigProb\) and the adoption level \(\penetration\), treating \(\accidentCost,\accidentFunc,\detectionFunc\) as fixed primitives. Therefore, we suppress this dependence and write \(\expP\profileInput\). When \(\sigProb\) and \(\penetration\) are fixed or clear from context, we simply write \(\expP(\equiProfile)\).

\section{Main Results}
\label{sec:main-results}
We now state the two main consequences of the corrected model. First, increased V2V adoption can increase equilibrium accident probability under fixed non-optimal signaling probabilities. Second, this perverse effect disappears when the signaling probability is chosen optimally at each adoption level.

\subsection{Perverse Effect of V2V Adoption}
Under the corrected model, we characterize the equilibrium profiles in Section~\ref{sec:proofs-region-characterization}. Let \(\expPNE\paraInput\) denote the equilibrium accident probability under this characterization. This quantity is uniquely determined for each \(\paraInput\), even when multiple equilibrium profiles exist. Our first main result shows that increased V2V adoption can increase equilibrium accident probability.
\begin{thm}
	\label{thm:perverse-effect-source}
	There exist model parameters satisfying the assumptions above and a fixed signaling probability \(\sigProb\in[0,1]\) such that the equilibrium accident probability
	\[
	\expPNE(\sigProb,\penetration)
	\]
	is increasing over an interval of V2V adoption levels \(\penetration\).
\end{thm}
The proof further shows that such a local increase can occur only through specific parameter regions, which are defined formally in Section~\ref{sec:proofs-region-characterization}. Fig.~\ref{fig:policy-accident-diff} plots \(\expPNE\) as a function of \(\penetration\) for fixed signaling probabilities \(\sigProb=0.05\) and \(\sigProb=0.35\). The corresponding heatmap in Fig.~\ref{fig:policy-accident-diff-heatmap}, generated using the same primitive parameters, shows that the increasing segments occur in certain regions. Although the increase in this illustrative instance is modest, it demonstrates the mechanism; larger increases can occur for other parameter values. Moreover, adverse adoption effects can occur at higher signaling probabilities under suitable model parameters.

\subsection{Optimal Signaling Removes the Perverse Effect}
One might expect that increasing V2V adoption always improves road safety. Theorem~\ref{thm:perverse-effect-source} shows that this need not be true under fixed non-optimal signaling probabilities. This raises a natural design question: can the system designer choose the signaling probability at each adoption level so that increased V2V adoption does not make the road less safe? Our second main result answers this question.
\begin{thm}
	\label{thm:no-perverse-optimal}
	For fixed primitives \(\accidentCost,\accidentFunc,\detectionFunc\), define an optimal signaling policy \(\optimalSigP:[0,1]\to[0,1]\) by selecting, for each \(\penetration\in[0,1]\),
	\[
	\optimalSigP(\penetration)
	\in
	\arg\min_{\sigProb\in[0,1]}\expPNE\paraInput.
	\]
	Then the optimized equilibrium accident probability
	\[
	\expPNE(\optimalSigP(\penetration),\penetration)
	\]
	is weakly decreasing in \(\penetration\). 
\end{thm}
Fig.~\ref{fig:policy-accident-diff} illustrates this result: the trace under \(\sigProb=\optimalSigP(\penetration)\) is weakly decreasing, while traces under fixed non-optimal signaling probabilities may exhibit increasing segments. The trace \(\sigProb=\optimalSigP(\penetration)\) in Fig.~\ref{fig:policy-accident-diff-heatmap} illustrates a switch from no signaling to full signaling as adoption increases. Proposition~\ref{prop:optimal-policy} provides a constructive characterization of one such optimal policy and implies that, under its assumptions and tie-breaking rule, this policy can switch at most once.
\begin{figure}[t]
	\begin{center}
		\includegraphics[width=8.4cm]{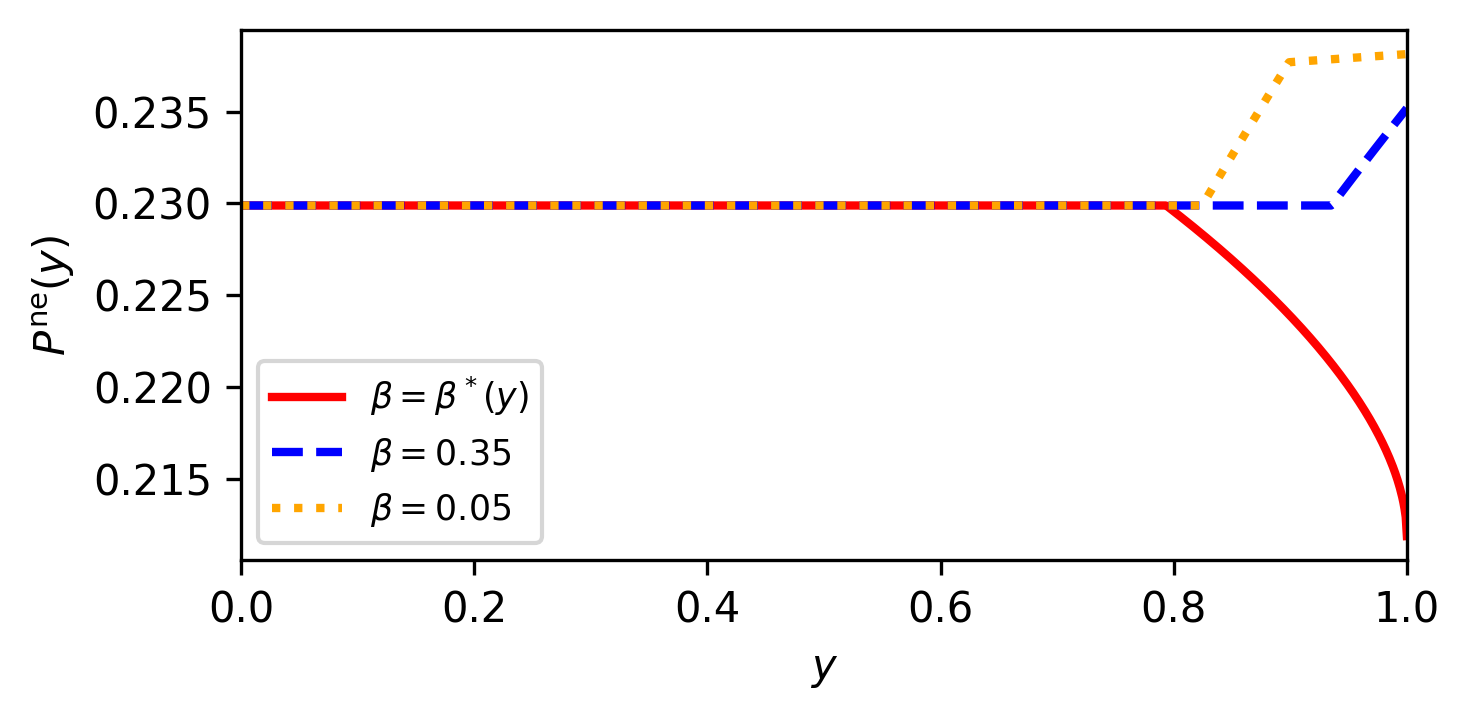}  
		\caption{Equilibrium accident probability as a function of V2V adoption under selected signaling choices, with \(\accidentCost=3.35\), \(\accidentFunc(z)=0.05+0.20\sqrt{z}\), and \(\detectionFunc(\penetration)=0.9\sqrt{\penetration}\). The traces under fixed non-optimal signaling probabilities exhibit increasing segments, while the optimized equilibrium accident probability is nonincreasing.} 
		\label{fig:policy-accident-diff}
	\end{center}
\end{figure}

\begin{table*}[t]
	\centering
	\caption{Equilibrium-region characterization.}
	\label{tb:region-characterization}
	\begin{tabular}{c|c|c|c|c|c}
		\toprule
		\(i\)
		& Region \(\regionSet_i\)
		& Threshold statement \(\Tstate_i\profileInput\)
		& Candidate set \(\candSet_i\)
		& Parameter condition \(\paraCond_i\paraInput\)
		& \(\expPNE\paraInput\) \\
		\midrule
		
		\(1\)
		& \(\regionSet_1\)
		& \(\expP\profileInput<\TOne\)
		& \(\{(1-\penetration,\penetration,0)\}\)
		& \(\expP_1\paraInput<\TOne\)
		& \(\expP_1\paraInput\)
		\\[4pt]
		
		\(2\)
		& \(\regionSet_2\)
		& \(\expP\profileInput=\TOne\)
		& \(\{(\equiProfile_\groupN,\penetration,0):\equiProfile_\groupN\in[0,1-\penetration]\}\)
		& \(\expP_2(0;\sigProb,\penetration)\le \TOne\le \expP_2(1-\penetration;\sigProb,\penetration)\)
		& \(\TOne\)
		\\[4pt]
		
		\(3\)
		& \(\regionSet_3\)
		& \(\TOne<\expP\profileInput<\TTwo\paraInput\)
		& \(\{(0,\penetration,0)\}\)
		& \(\TOne<\expP_3\paraInput<\TTwo\paraInput\)
		& \(\expP_3\paraInput\)
		\\[4pt]
		
		\(4\)
		& \(\regionSet_4\)
		& \(\expP\profileInput=\TTwo\paraInput\)
		& \(\{(0,\equiProfile_\groupVU,0):\equiProfile_\groupVU\in[0,\penetration]\}\)
		& \(\expP_4(0;\sigProb,\penetration)\le \TTwo\paraInput\le \expP_4(\penetration;\sigProb,\penetration)\)
		& \(\TTwo\paraInput\)
		\\[4pt]
		
		\(5\)
		& \(\regionSet_5\)
		& \(\TTwo\paraInput<\expP\profileInput\)
		& \(\{(0,0,0)\}\)
		& \(\TTwo\paraInput<\expP_5\)
		& \(\expP_5=\accidentFunc(0)\)
		\\
		\bottomrule
	\end{tabular}
\end{table*}
\section{Proofs}
\label{sec:proofs}
By \eqref{eq:old-new-consistency-diff}, the corrected consistency equation generally differs from the original one when \(\accidentFunc\) is nonlinear. Therefore, we establish the equilibrium characterization and endpoint optimality under the corrected model. The proofs are organized in three parts:
\begin{itemize}
	\item First, we verify the well-definedness of \eqref{eq:consistency-equation-explicit} and characterize the equilibrium outcomes generated by the corrected consistency equation.
	\item Second, we analyze the designer's optimization problem over the signaling probability and construct one optimal signaling policy.
	\item Finally, we prove Theorems~\ref{thm:perverse-effect-source} and~\ref{thm:no-perverse-optimal} using these intermediate results.
\end{itemize}

\begin{figure}[t]
	\begin{center}
		\includegraphics[width=8.4cm]{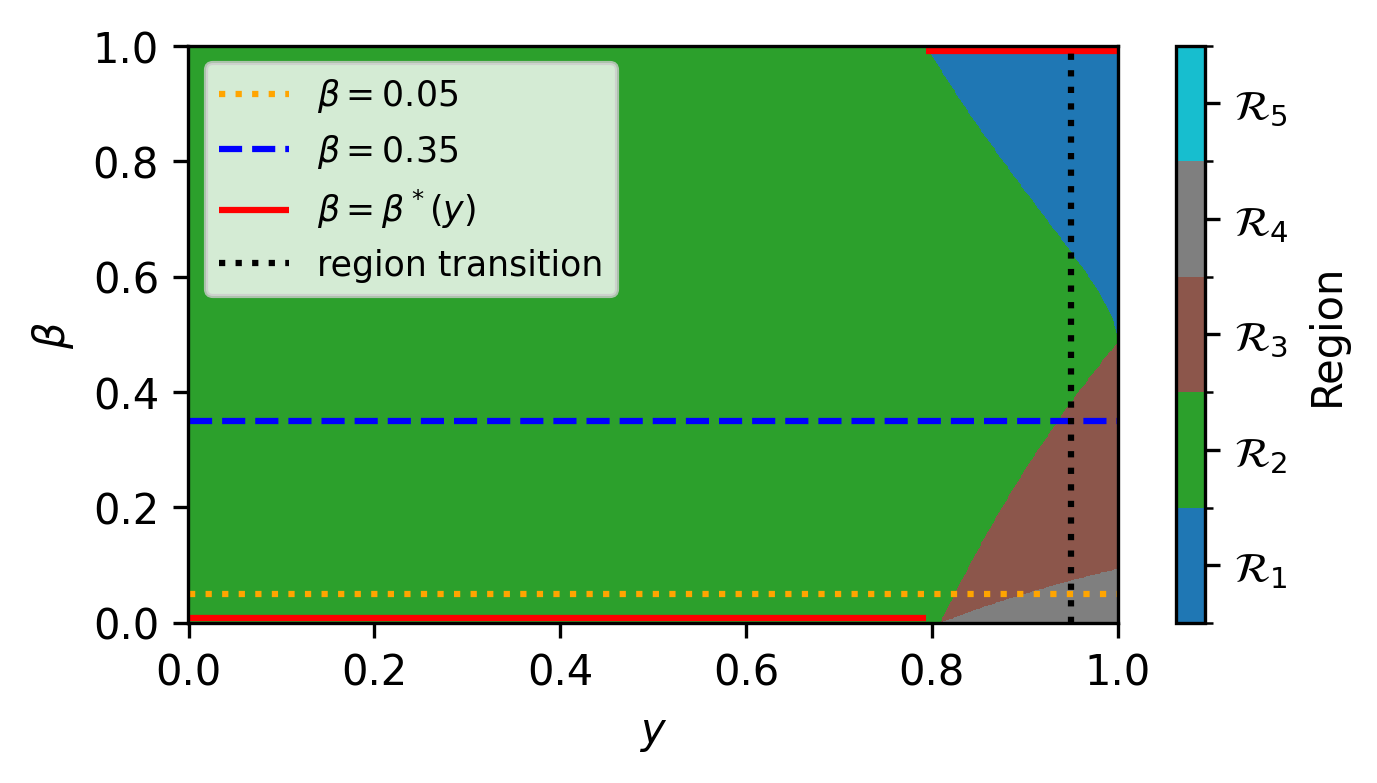}
		\caption{Parameter-region heatmap over \((\sigProb,\penetration)\) using the same model primitives as in Fig.~\ref{fig:policy-accident-diff}. The increasing segments in Fig.~\ref{fig:policy-accident-diff} occur when traces under fixed non-optimal signaling probabilities pass through \(\regionSet_3\) and/or \(\regionSet_4\), while the optimal signaling policy switches to \(\sigProb=1\) only in \(\regionSet_1\). For fixed \(\penetration\), the vertical region transitions illustrate the ordering of regions as \(\sigProb\) increases.}
		\label{fig:policy-accident-diff-heatmap}
	\end{center}
\end{figure}
\subsection{Equilibrium Region Characterization}
\label{sec:proofs-region-characterization}
We first collect the notation needed to state the equilibrium region characterization. The two \emph{behavioral thresholds} are defined by
\[
\TOne=\TOneVal,
\qquad
\TTwo\paraInput=\TTwoVal.
\]
These are the accident probabilities at which the costs of reckless and careful driving are equal for non-V2V drivers and unsignaled V2V drivers, respectively. We also define the following candidate accident probabilities:
\[
\expP_1\paraInput=\candProbOneVal,
\]
\[
\expP_2(\equiProfile_\groupN;\sigProb,\penetration)=\candProbTwoVal,
\]
\[
\expP_3\paraInput=\candProbThreeVal,
\]
\[
\expP_4(\equiProfile_\groupVU;\sigProb,\penetration)=\candProbFourVal,
\]
and
\[
\expP_5=\candProbFiveVal.
\]
Table~\ref{tb:region-characterization} collects the notation used for the equilibrium-region characterization. Before the characterization is proved, the table serves as a definition table; after Proposition~\ref{prop:region-characterization}, it becomes the equilibrium characterization table. For each \(i=1,\dots,5\), \(\Tstate_i\profileInput\) is a statement about where the induced accident probability \(\expP\profileInput\) lies relative to the two behavioral thresholds. The set \(\candSet_i\) describes the candidate profile form associated with that threshold case. These sets are called candidate sets because they describe possible equilibrium forms; in strict regions, \(\candSet_i\) is a singleton, while in equality regions, not every member of \(\candSet_i\) is necessarily an equilibrium. The condition \(\paraCond_i\paraInput\) is the parameter-side statement associated with the \(i\)-th threshold case; it is what allows us to determine from \(\paraInput\) alone which parameter region is realized. The associated parameter region is
\[
\regionSet_i := \{\paraInput: \paraCond_i\paraInput\text{ holds}\}.
\]
Thus \(\regionSet_i\) is a subset of \(\paraInput\)-space. This matches region-colored heatmaps, such as Fig.~\ref{fig:policy-accident-diff-heatmap}: each point \(\paraInput\) belongs to one of the parameter regions \(\regionSet_i\). Finally, \(\expPNE\paraInput\) denotes the resulting equilibrium accident probability.

For fixed \(\paraInput\), define
\[
\NE\paraInput
:=
\{\equiProfile:\ \equiProfile \text{ is an equilibrium under }\paraInput\}.
\]
With this notation in place, we can state the equilibrium-region characterization compactly.
\begin{prop}
	\label{prop:region-characterization}
	For fixed \(\paraInput\), the equilibrium characterization is summarized by Table~\ref{tb:region-characterization}. For each \(i\in\{1,\dots,5\}\), first, if \(\sigProb\detectionFunc(\penetration) > 0\), then
	\[
	\equiProfile\in\NE\paraInput
	\text{ and }
	\Tstate_i\profileInput
	\implies
	\equiProfile\in\candSet_i.
	\]
	Thus \(\candSet_i\) exhausts the possible equilibrium profile forms in behavioral threshold case \(\Tstate_i\) when \(\sigProb\detectionFunc(\penetration) > 0\). Second,
	\[
	\paraCond_i\paraInput
	\iff
	\exists \equiProfile\in\candSet_i
	\text{ s.t. }
	\equiProfile\in\NE\paraInput
	\text{ and }
	\Tstate_i\profileInput.
	\]
	Equivalently, \(\regionSet_i\) is the set of parameter values for which some equilibrium realizes behavioral threshold case \(\Tstate_i\). The final column of Table~\ref{tb:region-characterization} gives the corresponding equilibrium accident probability.
\end{prop}
Before proving the characterization, we verify that the corrected consistency equation \eqref{eq:consistency-equation-explicit} behaves properly on the candidate profile sets appearing in Table~\ref{tb:region-characterization}.
\begin{lem}
	\label{lem:boundedness-of-consistency-equation}
	For every candidate profile \(\equiProfile\in\candSet_i\), \(i=1,\dots,5\), the explicit consistency expression \eqref{eq:consistency-equation-explicit} is well-defined and satisfies \(\expP\profileInput\in[0,1]\). Moreover,
	\[
	\accidentFunc(0)
	\leq
	\expP\profileInput
	\leq
	\accidentFunc(1).
	\]
\end{lem}
\begin{pf}[Lemma~\ref{lem:boundedness-of-consistency-equation}]
	The candidate profile sets in Table~\ref{tb:region-characterization} all satisfy \(\equiProfile_\groupVS=0\).
	Therefore, for every \(\equiProfile\in\candSet_i\), \(i=1,\dots,5\), we have
	\[
	\du
	=
	\equiProfile_\groupN+\equiProfile_\groupVU
	\ge
	\equiProfile_\groupN+\equiProfile_\groupVS
	=
	\ds.
	\]
	By weak monotonicity of \(\accidentFunc\), \(\accidentFunc(\du)\ge \accidentFunc(\ds)\).
	Hence the denominator in \eqref{eq:consistency-equation-explicit} satisfies \[
	1+\sigProb\detectionFunc(\penetration)
	\bigl(\accidentFunc(\du)-\accidentFunc(\ds)\bigr)
	\ge 1.
	\]
	Thus the explicit expression is well-defined on every candidate profile set. Moreover,
	\[
	0
	\le
	\expP\profileInput
	=
	\frac{\accidentFunc(\du)}
	{1+\sigProb\detectionFunc(\penetration)
		\bigl(\accidentFunc(\du)-\accidentFunc(\ds)\bigr)}
	\le
	\accidentFunc(\du)
	\le 1.
	\]
	Therefore, \(\expP\profileInput\in[0,1]\). From \eqref{eq:signal-display-probability}, it follows that the signal-display probability is also valid, \(\probability{\eventDisp}\in [0, 1]\). Now using the implicit consistency equation \eqref{eq:consistency-equation-new}, we have
	\[
	\expP\profileInput
	=
	\bigl(1-\probability{\eventDisp}\bigr)\accidentFunc(\du)
	+
	\probability{\eventDisp}\accidentFunc(\ds).
	\]
	Since \(\probability{\eventDisp}\in[0,1]\), this is a convex combination of \(\accidentFunc(\du)\) and \(\accidentFunc(\ds)\). Hence
	\[
	\min\{\accidentFunc(\du),\accidentFunc(\ds)\}
	\le
	\expP\profileInput
	\le
	\max\{\accidentFunc(\du),\accidentFunc(\ds)\}.
	\]
	Because \(\du,\ds\in[0,1]\) and \(\accidentFunc\) is weakly increasing,
	\[
	\accidentFunc(0)
	\le
	\expP\profileInput
	\le
	\accidentFunc(1).
	\]
	This completes the proof.
\end{pf}
Next, we prove a monotonicity property for the candidate accident probabilities \(\expP_2(\equiProfile_\groupN)\) and \(\expP_4(\equiProfile_\groupVU)\).
\begin{lem}
	\label{lem:candidate-probability-monotonicity}
	Fix \(\paraInput\). The functions \(\expP_2(\equiProfile_\groupN)\) and \(\expP_4(\equiProfile_\groupVU)\) are weakly increasing in \(\equiProfile_\groupN\in[0,1-\penetration]\) and \(\equiProfile_\groupVU\in[0,\penetration]\), respectively.
\end{lem}
\begin{pf}[Lemma~\ref{lem:candidate-probability-monotonicity}]
	Fix \(\paraInput\), and let \(\effSig=\sigProb\detectionFunc(\penetration)\). Define
	\[
	F(\utilityU,\utilityV):=\frac{\utilityU}{1+\effSig(\utilityU-\utilityV)}.
	\]
	On the candidate domain considered below, we have \(\utilityU\geq \utilityV\), so the denominator is at least \(1\). We first show how \(F\) changes when one argument is varied while the other is fixed. For fixed \(\utilityV\), suppose \(\utilityU_2\ge \utilityU_1\ge \utilityV\). Then
	\[
	F(\utilityU_2,\utilityV)-F(\utilityU_1,\utilityV)
	=
	\frac{(\utilityU_2-\utilityU_1)(1-\effSig \utilityV)}
	{\bigl(1+\effSig(\utilityU_2-\utilityV)\bigr)
		\bigl(1+\effSig(\utilityU_1-\utilityV)\bigr)}
	\geq 0.
	\]
	Thus \(F\) is weakly increasing in \(\utilityU\) whenever the fixed \(\utilityV\) is no larger than the smaller \(\utilityU\). For fixed \(\utilityU\), if \(\utilityV_2\geq \utilityV_1\), then
	\[
	1+\effSig(\utilityU-\utilityV_2)
	\leq
	1+\effSig(\utilityU-\utilityV_1).
	\]
	Hence \(F(\utilityU,\utilityV_2)\ge F(\utilityU,\utilityV_1)\), so \(F\) is weakly increasing in \(\utilityV\). Chaining these two facts in the safe order, if \(\utilityU_2\geq \utilityU_1\geq \utilityV_1\) and \(\utilityV_2\geq \utilityV_1\), then
	\[
	F(\utilityU_2,\utilityV_2)
	\geq
	F(\utilityU_2,\utilityV_1)
	\geq
	F(\utilityU_1,\utilityV_1).
	\]
	For \(\expP_2\), take \(\utilityU=\accidentFunc(\equiProfile_\groupN+\penetration)\) and \(\utilityV=\accidentFunc(\equiProfile_\groupN)\). If \(\equiProfile_{\groupN,2}\geq \equiProfile_{\groupN,1}\), then weak monotonicity of \(\accidentFunc\) gives \(\utilityU_2\geq \utilityU_1\) and \(\utilityV_2\geq \utilityV_1\). Also, since \(\equiProfile_{\groupN,1}+\penetration\geq\equiProfile_{\groupN,1}\), we have \(\utilityU_1\geq\utilityV_1\). Therefore the chained inequality above implies
	\[
	\expP_2(\equiProfile_{\groupN,2};\sigProb,\penetration)
	\geq
	\expP_2(\equiProfile_{\groupN,1};\sigProb,\penetration).
	\]
	Thus \(\expP_2\) is weakly increasing in \(\equiProfile_\groupN\). The argument for \(\expP_4\) is identical by taking \(\utilityU=\accidentFunc(\equiProfile_\groupVU)\) and \(\utilityV=\accidentFunc(0)\), so \(\expP_4\) is weakly increasing in \(\equiProfile_\groupVU\), and this completes the proof.
\end{pf}
We next record several threshold relations used throughout the equilibrium characterization proof. Since \(\sigProb\detectionFunc(\penetration)\in[0,1]\), we have
\begin{equation}
	\label{eq:threshold-relations}
	\TOne=\TOneVal\leq \TTwoVal=\TTwo\paraInput.
\end{equation}
Moreover, using \eqref{eq:signal-display-probability}, whenever \(\probability{\eventDispNeg}>0\), we have
\[
\condProb{\eventAcci}{\eventDispNeg}
=
\frac{
	\expP\profileInput(1-\sigProb\detectionFunc(\penetration))
}{
	1-\expP\profileInput\sigProb\detectionFunc(\penetration)
}.
\]
Therefore,
\begin{equation}
	\label{eq:threshold-relations-equivalence}
	\condProb{\eventAcci}{\eventDispNeg}
	\lesseqqgtr
	\TOne
	\iff
	\expP\profileInput
	\lesseqqgtr
	\TTwo\paraInput.
\end{equation}
We now combine the consistency-equation properties and threshold relations to prove the equilibrium-region characterization.
\begin{pf}[Proposition~\ref{prop:region-characterization}]
	The proof is row-wise. We first discuss the boundary case \(\sigProb\detectionFunc(\penetration)=0\). In this case, \(\TOne=\TTwo\paraInput\), so the two equality threshold cases coincide and the open interval case
	\[
	\TOne<\expP\profileInput<\TTwo\paraInput
	\]
	is empty. Thus the five-row characterization collapses to the corresponding three threshold possibilities. This does not affect the equilibrium accident probability, since the two equality rows assign the same value
	\[
	\expPNE\paraInput=\TOne=\TTwo\paraInput.
	\]
	When \(\sigProb\detectionFunc(\penetration)=0\), non-V2V drivers and unsignaled V2V drivers face the same behavioral threshold, so the two equality cases merge. In this boundary case, the candidate sets \(\candSet_2\) and \(\candSet_4\) provide representative rather than exhaustive equilibrium profile forms, since both groups may be indifferent. Therefore, for the representative proof below, we focus on the non-boundary case \(\sigProb\detectionFunc(\penetration)>0\), where \(\TOne<\TTwo\paraInput\) by \eqref{eq:threshold-relations}. Under this condition, the five behavioral threshold cases are mutually exclusive and exhaustive: every induced accident probability must fall below, equal to, between, equal to, or above the two thresholds \(\TOne\) and \(\TTwo\paraInput\).
	
	We prove the representative equality case \(i=4\) in detail. The argument for \(i=2\) is analogous, with \(\expP_2\), \(\equiProfile_\groupN\), and \(\TOne\) replacing \(\expP_4\), \(\equiProfile_\groupVU\), and \(\TTwo\paraInput\). The strict cases \(i=1,3,5\) follow by the same logic, with the simplification that their candidate sets are singletons. First, we show 
	\[
	\equiProfile\in\NE\paraInput \text{ and } \Tstate_4\profileInput \implies \equiProfile\in\candSet_4.
	\]
	By \eqref{eq:threshold-relations-equivalence} and Table~\ref{tb:expected-cost}, \(\Tstate_4\profileInput\) implies \(\expC_\groupVU(\equiProfile;\actionCare)=\expC_\groupVU(\equiProfile;\actionRec)\). Hence any reckless mass \(\equiProfile_\groupVU\in[0,\penetration]\) is compatible with the equilibrium inequalities for this group; the exact mass is not pinned down by the cost comparison alone. Since \(\sigProb\detectionFunc(\penetration)>0\), \(\Tstate_4\profileInput\) and \eqref{eq:threshold-relations} imply \(\expP\profileInput=\TTwo\paraInput>\TOne\). Rearranging this inequality and using Table~\ref{tb:expected-cost}, we obtain \(\expC_\groupN(\equiProfile;\actionRec)>\expC_\groupN(\equiProfile;\actionCare)\). Since \(\equiProfile\in\NE\paraInput\), the contrapositive of \eqref{eq:exp-cost-reckless-optimal} gives \(\equiProfile_\groupN=0\). Since \(\condProb{\eventAcci}{\eventDisp}=1\), we always have \(\expC_\groupVS(\equiProfile;\actionRec)>\expC_\groupVS(\equiProfile;\actionCare)\), so the contrapositive of \eqref{eq:exp-cost-reckless-optimal} gives \(\equiProfile_\groupVS=0\). Therefore any equilibrium satisfying \(\Tstate_4\profileInput\) has the candidate form
	\[
	\equiProfile=(0,\equiProfile_\groupVU,0),
	\qquad
	\equiProfile_\groupVU\in[0,\penetration],
	\]
	so \(\equiProfile\in\candSet_4\), and \(\candSet_4\) exhausts all possible equilibrium profiles in this region. It remains to prove the parameter-condition equivalence for this row:
	\[
	\paraCond_4\paraInput
	\iff
	\exists \equiProfile\in\candSet_4
	\text{ s.t. }
	\equiProfile\in\NE\paraInput
	\text{ and }
	\Tstate_4\profileInput.
	\]
	Suppose first that there exists \(\equiProfile\in\candSet_4\) such that \(\equiProfile\in\NE\paraInput\) and \(\Tstate_4\profileInput\). Since \(\equiProfile\in\candSet_4\), we can write \(\equiProfile=(0,\equiProfile_\groupVU,0)\) for some \(\equiProfile_\groupVU\in[0,\penetration]\). Substituting this form into the explicit consistency equation \eqref{eq:consistency-equation-explicit} gives
	\[
	\expP\profileInput
	=
	\expP_4(\equiProfile_\groupVU;\sigProb,\penetration).
	\]
	Because \(\Tstate_4\profileInput\) holds,
	\[
	\expP_4(\equiProfile_\groupVU;\sigProb,\penetration)
	=
	\expP\profileInput
	=
	\TTwo\paraInput.
	\]
	By Lemma~\ref{lem:candidate-probability-monotonicity}, \(\expP_4\) is weakly increasing in \(\equiProfile_\groupVU\). Therefore
	\[
	\expP_4(0;\sigProb,\penetration)
	\le
	\TTwo\paraInput
	\le
	\expP_4(\penetration;\sigProb,\penetration),
	\]
	which is exactly \(\paraCond_4\paraInput\). Conversely, suppose \(\paraCond_4\paraInput\) holds. Since \(\accidentFunc\) is continuous, \(\expP_4\) is continuous in \(\equiProfile_\groupVU\). Therefore, by the Intermediate Value Theorem, there exists \(\equiProfile_\groupVU^*\in[0,\penetration]\) such that
	\[
	\expP_4(\equiProfile_\groupVU^*;\sigProb,\penetration)
	=
	\TTwo\paraInput.
	\]
	Construct \(\equiProfile^*:=(0,\equiProfile_\groupVU^*,0)\in\candSet_4\). Then
	\[
	\expP(\equiProfile^*;\sigProb,\penetration)
	=
	\expP_4(\equiProfile_\groupVU^*;\sigProb,\penetration)
	=
	\TTwo\paraInput,
	\]
	so \(\Tstate_4(\equiProfile^*;\sigProb,\penetration)\) holds. Using \(\Tstate_4(\equiProfile^*;\sigProb,\penetration)\), \eqref{eq:threshold-relations}, \eqref{eq:threshold-relations-equivalence}, and Table~\ref{tb:expected-cost}, we obtain \(\expC_\groupVU(\equiProfile^*;\actionCare)=\expC_\groupVU(\equiProfile^*;\actionRec)\), \(\expC_\groupN(\equiProfile^*;\actionRec)>\expC_\groupN(\equiProfile^*;\actionCare)\), and \(\expC_\groupVS(\equiProfile^*;\actionRec)>\expC_\groupVS(\equiProfile^*;\actionCare)\). Since \(\equiProfile^*=(0,\equiProfile_\groupVU^*,0)\), the constructed profile satisfies the equilibrium inequalities \eqref{eq:exp-cost-reckless-optimal} and \eqref{eq:exp-cost-careful-optimal}. Thus \(\equiProfile^*\in\NE\paraInput\), completing the proof.
\end{pf}

\subsection{Optimal Signaling}
\label{sec:proofs-optimal-signaling-policy}
We next turn from equilibrium characterization to the designer's problem of choosing the signaling probability. The next proposition gives one explicit optimal signaling policy, using a tie-breaking rule in favor of no signaling.
\begin{prop}
	\label{prop:optimal-policy}
	Suppose \(\accidentFunc\) and \(\detectionFunc\) are strictly increasing. With ties broken in favor of no signaling, one such optimal signaling policy \(\optimalSigP\) can be selected as
	\begin{equation}
		\label{eq:an-optimal-signaling-policy}
		\optimalSigP(\penetration)
		=
		\begin{cases}
			1, & \expP_1(1,\penetration) < \TOne,\ \penetration>0,\ \detectionFunc(\penetration)>0,\\[3pt]
			0, & \text{otherwise}.
		\end{cases}
	\end{equation}
\end{prop}
The proof of the proposition proceeds in three steps: first we show how regions are ordered as the signaling probability varies, then we prove endpoint optimality, and finally we compare the optimized endpoint values. The first step is to show that, for fixed adoption level, the realized equilibrium region can move only in a particular order as the signaling probability increases.
\begin{lem}
	\label{lem:region-ordering}
	Fix \(\penetration\). As \(\sigProb\) increases, the realized equilibrium regions can move only through a contiguous subchain of
	\begin{equation}
		\label{eq:region-ordering}
		\regionSet_5
		\to
		\regionSet_4
		\to
		\regionSet_3
		\to
		\regionSet_2
		\to
		\regionSet_1.
	\end{equation}
	In particular, the transition cannot move backward or skip an intermediate region.
\end{lem}
\begin{pf}[Lemma~\ref{lem:region-ordering}]
	If \(\detectionFunc(\penetration)=0\), then \(\effSig=\sigProb\detectionFunc(\penetration)=0\) for all \(\sigProb\), so the realized region does not change with \(\sigProb\). Hence the realized path is a one-element subchain. We therefore assume \(\detectionFunc(\penetration)>0\). Then \(\effSig=\sigProb\detectionFunc(\penetration)\) increases with \(\sigProb\). Using the endpoint descriptions of \(\regionSet_2\) and \(\regionSet_4\), the equality-region conditions can be written as
	\[
	\regionSet_4:\quad \expP_5\leq \TTwo\paraInput\leq \expP_3\paraInput,
	\]
	\[
	\regionSet_2:\quad \expP_3\paraInput\leq \TOne\leq \expP_1\paraInput.
	\]
	Moreover, by Lemma~\ref{lem:candidate-probability-monotonicity}, the same endpoint relations also give the global ordering
	\[
	\expP_5\leq \expP_3\paraInput\leq \expP_1\paraInput.
	\]
	As \(\effSig\) increases, \(\TTwo\paraInput\) increases, while \(\expP_3\paraInput\) and \(\expP_1\paraInput\) weakly decrease; \(\TOne\) and \(\expP_5\) are fixed. This ordering is illustrated in Fig.~\ref{fig:p-line}. Hence the relevant crossing conditions can only switch in the order
	\[
	\TTwo=\expP_5
	\;\;\to\;\;
	\TTwo=\expP_3
	\;\;\to\;\;
	\expP_3=\TOne
	\;\;\to\;\;
	\expP_1=\TOne.
	\]
	These crossings correspond to the transitions
	\[
	\regionSet_5\to\regionSet_4,
	\qquad
	\regionSet_4\to\regionSet_3,
	\qquad
	\regionSet_3\to\regionSet_2,
	\qquad
	\regionSet_2\to\regionSet_1.
	\]
	These transitions can occur only in the forward direction. For example, once \(\TTwo\paraInput\geq \expP_5\), the system cannot return to \(\regionSet_5\), because \(\TTwo\paraInput\) increases with \(\effSig\) and \(\expP_5\) is fixed. Similarly, once \(\TTwo\paraInput\geq \expP_3\paraInput\), the system cannot return to \(\regionSet_4\), because \(\TTwo\paraInput\) increases while \(\expP_3\paraInput\) weakly decreases. The same reasoning applies to the crossings \(\expP_3\paraInput=\TOne\) and \(\expP_1\paraInput=\TOne\). Equality may persist over an interval of \(\sigProb\), but the ordering cannot reverse. Thus the realized path cannot move backward.
	
	Finally, the fixed ordering \(\expP_5\leq\expP_3\paraInput\leq\expP_1\paraInput\) prevents any intermediate region from being skipped. The crossing \(\TTwo\paraInput=\expP_3\paraInput\) cannot occur before \(\TTwo\paraInput=\expP_5\), because \(\expP_5\leq\expP_3\paraInput\). Similarly, \(\expP_3\paraInput=\TOne\) cannot occur before \(\TTwo\paraInput=\expP_3\paraInput\), because \(\TOne\leq\TTwo\paraInput\); and \(\expP_1\paraInput=\TOne\) cannot occur before \(\expP_3\paraInput=\TOne\), because \(\expP_3\paraInput\leq\expP_1\paraInput\). Therefore the realized region path can only be a contiguous subchain of \eqref{eq:region-ordering}, and this completes the proof.
\end{pf}
\begin{figure}[t]
	\centering
	\resizebox{8.4cm}{!}{
		\begin{tikzpicture}[>=Stealth]
			
			\draw[->] (0,0) -- (10,0) node[right] {$\expP$};
			
			\draw[thick] (0,0.12) -- (0,-0.12);
			\node[below, font=\small] at (0,-0.12) {$0$};
			\draw[thick] (9.7,0.12) -- (9.7,-0.12);
			\node[below, font=\small] at (9.7,-0.12) {$1$};
			
			\draw[thick] (2,0.12) -- (2,-0.12);
			\node[below, font=\small] at (2,-0.12) {$\expP_5=\accidentFunc(0)$};
			\draw[thick] (9,0.12) -- (9,-0.12);
			\node[below, font=\small] at (9,-0.12) {$\accidentFunc(1)$};
			
			\draw[thick] (4,0.12) -- (4,-0.12);
			\node[below, font=\small] at (4,-0.12) {$\TOne$};
			
			\draw[thick] (5.5,0.12) -- (5.5,-0.12);
			\node[above, font=\small] (Ttwo) at (5.5,0.12) {$\TTwo(\effSig)$};
			\draw[->] ($(Ttwo.north)+(-0.5,0.15)$)--($(Ttwo.north)+(0.5,0.15)$);
			\node[above, font=\small] at ($(Ttwo.north)+(0,0.15)$) {$\uparrow \effSig$};
			
			
			\draw[thick] (7,0.12) -- (7,-0.12);
			\node[above, font=\small] (P3) at (7,0.12) {$\expP_3(\effSig)$};
			\draw[<-] ($(P3.north)+(-0.5,0.15)$)--($(P3.north)+(0.5,0.15)$);
			\node[above, font=\small] at ($(P3.north)+(0,0.15)$) {$\uparrow \effSig$};
			
			\draw[thick] (8.5,0.12) -- (8.5,-0.12);
			\node[above, font=\small] (P1) at (8.5,0.12) {$\expP_1(\effSig)$};
			\draw[<-] ($(P1.north)+(-0.5,0.15)$)--($(P1.north)+(0.5,0.15)$);
			\node[above, font=\small] at ($(P1.north)+(0,0.15)$) {$\uparrow \effSig$};
			
			\draw[decorate, decoration={brace, mirror, amplitude=4pt}]
			(2.2,-0.75) -- (3.8,-0.75)
			node[midway, below=6pt, font=\small] {$\regionSet_1$};
			
			\fill (4,-0.75) circle (1.5pt);
			\node[below=5pt, font=\small] at (4,-0.75) {($\regionSet_2$)};
			
			\draw[decorate, decoration={brace, mirror, amplitude=4pt}]
			(4.2,-0.75) -- (5.3,-0.75)
			node[midway, below=6pt, font=\small] {$\regionSet_3$};
			
			\fill (5.5,-0.75) circle (1.5pt);
			\node[below=5pt, font=\small] at (5.5,-0.75) {($\regionSet_4$)};
			
			\draw[decorate, decoration={brace, mirror, amplitude=4pt}]
			(5.7,-0.75) -- (8.8,-0.75)
			node[midway, below=6pt, font=\small] {$\regionSet_5$};
		\end{tikzpicture}
	}
	\caption{Schematic illustrating the relative positions of behavioral thresholds and candidate accident probabilities as \(\effSig=\sigProb\detectionFunc(\penetration)\) increases. Only the orderings \(\expP_5 \leq \expP_3\leq \expP_1\) and \(\TOne\leq \TTwo\) are guaranteed. The relative positions of quantities belonging to different orderings are illustrative and depend on the model parameters.}
	\label{fig:p-line}
\end{figure}
This region ordering rules out interior optima and leads to endpoint optimality.
\begin{lem}
	\label{lem:endpoint-optimality}
	Fix \(\penetration\in[0,1]\). The equilibrium accident probability is minimized by an endpoint signaling probability. That is,
	\[
	\min_{\sigProb\in[0,1]}\expPNE\paraInput
	=
	\min\{\expPNE(0,\penetration),\expPNE(1,\penetration)\}.
	\]
\end{lem}
\begin{pf}[Lemma~\ref{lem:endpoint-optimality}]
	Fix \(\penetration\in[0,1]\). By Lemma~\ref{lem:region-ordering}, as \(\sigProb\) increases, the realized equilibrium region can move only through a contiguous subchain \eqref{eq:region-ordering}. By Proposition~\ref{prop:region-characterization}, the corresponding equilibrium accident probabilities are
	\[
	\expP_5,\qquad
	\TTwo\paraInput,\qquad
	\expP_3\paraInput,\qquad
	\TOne,\qquad
	\expP_1\paraInput.
	\]
	As \(\sigProb\) increases, \(\expP_5\) and \(\TOne\) are constant, \(\TTwo\paraInput\) weakly increases, and \(\expP_3\paraInput\) and \(\expP_1\paraInput\) weakly decrease. Therefore, along any realized subchain, \(\expPNE\paraInput\) can only be constant or increasing before reaching \(\regionSet_3\), and constant or decreasing after leaving \(\regionSet_4\). Thus the minimum of \(\expPNE\paraInput\) over \(\sigProb\in[0,1]\) must occur at one of the endpoints.
\end{pf}
Endpoint optimality reduces the designer's problem to comparing the two endpoint signaling probabilities, \(\sigProb=0\) and \(\sigProb=1\).
\begin{lem}
	\label{lem:optimized-equilibrium-accident-probability}
	For each \(\penetration\in[0,1]\), let \(\optimalSigP\) be any optimal signaling policy as defined in Theorem~\ref{thm:no-perverse-optimal}. Then the optimized equilibrium accident probability is
	\begin{equation}
		\label{eq:optimized-equilibrium-accident-probability}
		\expPNE(\optimalSigP(\penetration),\penetration)
		=
		\begin{cases}
			\accidentFunc(0), & \TOne<\accidentFunc(0),\\[3pt]
			\min\{\TOne,\expP_1(1,\penetration)\},
			& \accidentFunc(0)\leq \TOne\leq \accidentFunc(1),\\[3pt]
			\expP_1(1,\penetration), & \accidentFunc(1)<\TOne.
		\end{cases}
	\end{equation}
\end{lem}
\begin{pf}[Lemma~\ref{lem:optimized-equilibrium-accident-probability}]
	Fix \(\penetration\in[0,1]\). By Lemma~\ref{lem:endpoint-optimality}, it is enough to compare the two endpoint signaling probabilities \(\sigProb=0\) and \(\sigProb=1\). When \(\sigProb=0\), we have \(\TTwo=\TOne\), so by Proposition~\ref{prop:region-characterization}, the no-signaling equilibrium accident probability is
	\[
	\expPNE(0,\penetration)
	=
	\begin{cases}
		\accidentFunc(0), & \TOne<\accidentFunc(0),\\[3pt]
		\TOne, & \accidentFunc(0)\leq \TOne\leq \accidentFunc(1),\\[3pt]
		\accidentFunc(1), & \accidentFunc(1)<\TOne.
	\end{cases}
	\]
	We compare the full-signaling endpoint against this benchmark. First suppose \(\TOne<\accidentFunc(0)\). By Lemma~\ref{lem:boundedness-of-consistency-equation}, every equilibrium accident probability satisfies \(\expPNE\paraInput\geq \accidentFunc(0)\). Hence
	\[
	\expPNE(1,\penetration)\geq \accidentFunc(0)=\expPNE(0,\penetration),
	\]
	so the optimized equilibrium accident probability is \(\accidentFunc(0)\). Next suppose \(\accidentFunc(0)\leq \TOne\leq \accidentFunc(1)\). Then \(\expPNE(0,\penetration)=\TOne\). By Proposition~\ref{prop:region-characterization}, the full-signaling endpoint has equilibrium accident probability below \(\TOne\) exactly when it lies in \(\regionSet_1\), equivalently when
	\[
	\expP_1(1,\penetration)<\TOne.
	\]
	Therefore the optimized equilibrium accident probability is
	\[
	\min\{\TOne,\expP_1(1,\penetration)\}.
	\]
	Finally suppose \(\accidentFunc(1)<\TOne\). By Lemma~\ref{lem:boundedness-of-consistency-equation}, every equilibrium accident probability satisfies \(\expPNE\paraInput\leq \accidentFunc(1)\). Hence
	\[
	\expPNE(1,\penetration)\leq \accidentFunc(1)=\expPNE(0,\penetration) < \TOne,
	\]
	so the optimized equilibrium accident probability is attained at \(\sigProb=1\). Since we also have \(\expPNE(1,\penetration)<\TOne\), the full-signaling endpoint lies in \(\regionSet_1\) in this case. Therefore, its equilibrium accident probability is
	\[
	\expPNE(1,\penetration)=\expP_1(1,\penetration).
	\]
	Combining the three cases proves the lemma.
\end{pf}
With the optimized endpoint value characterized, we now verify that the policy in Proposition~\ref{prop:optimal-policy} attains this optimized value under the stated tie-breaking convention.
\begin{pf}[Proposition~\ref{prop:optimal-policy}]
	Fix \(\penetration\in[0,1]\). By Lemma~\ref{lem:optimized-equilibrium-accident-probability}, the optimized equilibrium accident probability is characterized by \eqref{eq:optimized-equilibrium-accident-probability}. We show that the rule in \eqref{eq:an-optimal-signaling-policy} selects an endpoint attaining this optimized value in each case, with ties broken in favor of \(\sigProb=0\).
	
	First suppose \(\TOne<\accidentFunc(0)\). By Lemma~\ref{lem:optimized-equilibrium-accident-probability}, the optimized equilibrium accident probability is \(\accidentFunc(0)\). Since the no-signaling endpoint satisfies \(\expPNE(0,\penetration)=\accidentFunc(0)\), no signaling is optimal. Moreover, in this case \(\expP_1(1,\penetration)\geq \accidentFunc(0)>\TOne\) by Lemma~\ref{lem:boundedness-of-consistency-equation}, so the first line of \eqref{eq:an-optimal-signaling-policy} cannot hold. Hence the rule selects \(\sigProb=0\), as desired.
	
	Next suppose \(\accidentFunc(0)\leq \TOne\leq \accidentFunc(1)\). In this case, the no-signaling endpoint gives
	\[
	\expPNE(0,\penetration)=\TOne.
	\]
	By Proposition~\ref{prop:region-characterization}, the full-signaling endpoint has equilibrium accident probability below \(\TOne\) exactly when the full-signaling endpoint lies in \(\regionSet_1\), equivalently when
	\[
	\expP_1(1,\penetration)<\TOne.
	\]
	Thus full signaling is strictly better than no signaling exactly under the inequality in the first line of \eqref{eq:an-optimal-signaling-policy}. If this inequality fails, then no signaling is weakly optimal, and the tie-breaking convention selects \(\sigProb=0\). Hence the rule in \eqref{eq:an-optimal-signaling-policy} selects an optimal endpoint in this case.
	
	Finally suppose \(\accidentFunc(1)<\TOne\). By Lemma~\ref{lem:optimized-equilibrium-accident-probability}, the optimized equilibrium accident probability is \(\expP_1(1,\penetration)\), which is attained by the full-signaling endpoint whenever signaling has a nontrivial effect. Under the strict-monotonicity assumptions, if \(\penetration>0\) and \(\detectionFunc(\penetration)>0\), then
	\[
	\expP_1(1,\penetration)
	=
	\frac{\accidentFunc(1)}
	{1+\detectionFunc(\penetration)
		\bigl(\accidentFunc(1)-\accidentFunc(1-\penetration)\bigr)}
	<
	\accidentFunc(1)
	<
	\TOne.
	\]
	Therefore the first line of \eqref{eq:an-optimal-signaling-policy} holds, and the rule selects \(\sigProb=1\). If instead \(\penetration=0\) or \(\detectionFunc(\penetration)=0\), signaling has no effect, so the two endpoints tie and the tie-breaking convention selects \(\sigProb=0\). This is exactly the second line of \eqref{eq:an-optimal-signaling-policy}.
	
	Thus, in all cases, \eqref{eq:an-optimal-signaling-policy} selects an optimal endpoint, with ties broken in favor of no signaling. Therefore, it defines one optimal signaling policy under the stated tie-breaking convention.
\end{pf}

\subsection{Proofs of Main Theorems}
We now use Proposition~\ref{prop:region-characterization} and Lemma~\ref{lem:optimized-equilibrium-accident-probability} to prove Theorems~\ref{thm:perverse-effect-source} and~\ref{thm:no-perverse-optimal}. Theorem~\ref{thm:perverse-effect-source} follows from the equilibrium-region characterization and the numerical instance used in Fig.~\ref{fig:policy-accident-diff}.
\begin{pf}[Theorem~\ref{thm:perverse-effect-source}]
	The existence claim follows from the parameter instance used in Fig.~\ref{fig:policy-accident-diff}, where \(\expPNE\) increases over a nontrivial interval of V2V adoption levels \(\penetration\) under a fixed signaling probability \(\sigProb\).
	
	We also prove the accompanying claim that any local increase must pass through \(\regionSet_3\) or \(\regionSet_4\). By Proposition~\ref{prop:region-characterization}, the equilibrium accident probability in each region is given by the final column of Table~\ref{tb:region-characterization}.
	
	We first rule out the regions in which the equilibrium accident probability cannot increase with \(\penetration\). In \(\regionSet_2\), we have
	\[
	\expPNE\paraInput=\TOne,
	\]
	which is constant in \(\penetration\). In \(\regionSet_5\), we have
	\[
	\expPNE\paraInput=\expP_5=\accidentFunc(0),
	\]
	which is also constant in \(\penetration\). In \(\regionSet_1\), we have
	\[
	\expPNE\paraInput=\expP_1\paraInput,
	\]
	which weakly decreases as \(\penetration\) increases by weak monotonicity of \(\accidentFunc\) and \(\detectionFunc\). Therefore, within \(\regionSet_1\), \(\regionSet_2\), and \(\regionSet_5\), increasing \(\penetration\) cannot increase the equilibrium accident probability.
	
	The only remaining regions are \(\regionSet_3\) and \(\regionSet_4\). In \(\regionSet_3\), the equilibrium accident probability is
	\[
	\expPNE\paraInput=\expP_3\paraInput,
	\]
	and in \(\regionSet_4\), it is
	\[
	\expPNE\paraInput=\TTwo\paraInput.
	\]
	Moreover, \(\TTwo\paraInput\) is weakly increasing in \(\penetration\). Thus any local increase must arise through \(\regionSet_3\) or \(\regionSet_4\). This completes the proof.
\end{pf}
Theorem~\ref{thm:no-perverse-optimal} follows from the optimized equilibrium accident probability characterized above.
\begin{pf}[Theorem~\ref{thm:no-perverse-optimal}]
	By Lemma~\ref{lem:optimized-equilibrium-accident-probability}, the optimized equilibrium accident probability is given by \eqref{eq:optimized-equilibrium-accident-probability}. The three cases in \eqref{eq:optimized-equilibrium-accident-probability} are determined by the fixed quantities \(\TOne\), \(\accidentFunc(0)\), and \(\accidentFunc(1)\), and therefore do not change with \(\penetration\). We show that the optimized equilibrium accident probability is weakly decreasing in \(\penetration\) in each case.
	
	First suppose \(\TOne<\accidentFunc(0)\). Then
	\[
	\expPNE(\optimalSigP(\penetration),\penetration)=\accidentFunc(0),
	\]
	which is constant in \(\penetration\). Next suppose \(\accidentFunc(0)\leq \TOne\leq \accidentFunc(1)\). Then
	\[
	\expPNE(\optimalSigP(\penetration),\penetration)
	=
	\min\{\TOne,\expP_1(1,\penetration)\}.
	\]
	Since \(\TOne\) is constant and \(\expP_1(1,\penetration)\) is weakly decreasing in \(\penetration\), their minimum is weakly decreasing. Finally suppose \(\accidentFunc(1)<\TOne\). Then
	\[
	\expPNE(\optimalSigP(\penetration),\penetration)=\expP_1(1,\penetration),
	\]
	which is weakly decreasing in \(\penetration\).
	
	Therefore, in all three fixed benchmark cases, the optimized equilibrium accident probability is weakly decreasing in \(\penetration\). Hence, under an optimal signaling policy, increasing V2V adoption does not increase the equilibrium accident probability, completing the proof.
\end{pf}

\section{Discussion}
\label{sec:discussion}
The results above show that signaling is useful only when the system can enter the safe region \(\regionSet_1\), where \(\expP_1(\sigProb,\penetration)<\TOne\). This condition is restrictive. Indeed, since \(\detectionFunc(\penetration)\le 1\) and \(\accidentFunc(1-\penetration)\ge \accidentFunc(0)\), we have
\begin{equation}
	\label{eq:attainability-of-effective-full-signaling}
	\frac{\accidentFunc(1)}
	{1+\accidentFunc(1)-\accidentFunc(0)}
	\le
	\expP_1(1,\penetration).
\end{equation}
Thus, even under the most favorable detection case, \(\expP_1(1,\penetration)\) cannot fall below the left-hand side of \eqref{eq:attainability-of-effective-full-signaling}. Therefore, for \(\regionSet_1\) to be attainable, it is necessary that
\[
\frac{\accidentFunc(1)}
{1+\accidentFunc(1)-\accidentFunc(0)}
<
\TOne
=
\frac{1}{1+\accidentCost}.
\]
Assuming \(\accidentFunc(1)>0\), this is equivalent to
\[
\accidentCost
<
\frac{1-\accidentFunc(0)}{\accidentFunc(1)}.
\]
Hence \(\regionSet_1\) is truly a ``safe'' region: for signaling to be useful, the expected accident cost must not be too large relative to the baseline accident-risk function. In this sense, accident warnings are beneficial only when the underlying accident-risk environment is sufficiently mild for behavioral adaptation to reduce accident risk.

\section{Conclusion}
\label{sec:conclusion}
This paper revisits a partial-adoption V2V signaling model under a corrected accident-probability consistency equation. We show that increased V2V adoption can have adverse effects under non-optimal signaling probabilities, but that these effects can be eliminated by choosing the signaling policy optimally. Several extensions remain for future work. First, allowing noisy warning signals would provide a more general model of V2V information reliability. Second, incorporating network structure or dynamic adoption would allow the model to capture spatial and temporal effects that are abstracted away in the present aggregate formulation.


\begin{ack}
This material is based upon work supported by the National Science Foundation under award number 2440836 and by NASA under award number 80NSSC25M7102.
This work is also supported in part by ONR grant N000142612120.
\end{ack}

\section*{DECLARATION OF GENERATIVE AI AND AI-ASSISTED TECHNOLOGIES IN THE WRITING PROCESS}
During the preparation of this work, the first author used ChatGPT in order to improve the language and readability. After using this tool, the authors reviewed and edited the content as needed and take full responsibility for the content of the publication.

\bibliography{ifacconf}             

@inproceedings{gould2022partial,
	title={On partial adoption of vehicle-to-vehicle communication: When should cars warn each other of hazards?},
	author={Gould, Brendan T and Brown, Philip N},
	booktitle={2022 American Control Conference (ACC)},
	pages={627--632},
	year={2022},
	organization={IEEE}
}

@article{gould2023information,
	title={Information design for Vehicle-to-Vehicle communication},
	author={Gould, Brendan T and Brown, Philip N},
	journal={Transportation research part C: emerging technologies},
	volume={150},
	pages={104084},
	year={2023},
	publisher={Elsevier}
}

@article{wardrop1952road,
	title={Road paper. some theoretical aspects of road traffic research.},
	author={Wardrop, John Glen},
	journal={Proceedings of the institution of civil engineers},
	volume={1},
	number={3},
	pages={325--362},
	year={1952},
	publisher={Thomas Telford-ICE Virtual Library}
}

@article{correa2004selfish,
	title={Selfish routing in capacitated networks},
	author={Correa, Jos{\'e} R and Schulz, Andreas S and Stier-Moses, Nicol{\'a}s E},
	journal={Mathematics of Operations Research},
	volume={29},
	number={4},
	pages={961--976},
	year={2004},
	publisher={INFORMS}
}

@article{dafermos1984some,
	title={On some traffic equilibrium theory paradoxes},
	author={Dafermos, Stella and Nagurney, Anna},
	journal={Transportation Research Part B: Methodological},
	volume={18},
	number={2},
	pages={101--110},
	year={1984},
	publisher={Elsevier}
}

@article{kamenica2011bayesian,
	title={Bayesian persuasion},
	author={Kamenica, Emir and Gentzkow, Matthew},
	journal={American Economic Review},
	volume={101},
	number={6},
	pages={2590--2615},
	year={2011},
	publisher={American Economic Association}
}

@article{bergemann2019information,
	title={Information design: A unified perspective},
	author={Bergemann, Dirk and Morris, Stephen},
	journal={Journal of Economic Literature},
	volume={57},
	number={1},
	pages={44--95},
	year={2019},
	publisher={American Economic Association 2014 Broadway, Suite 305, Nashville, TN 37203-2425}
}

@article{acemoglu2018informational,
	title={Informational Braess’ paradox: The effect of information on traffic congestion},
	author={Acemoglu, Daron and Makhdoumi, Ali and Malekian, Azarakhsh and Ozdaglar, Asuman},
	journal={Operations Research},
	volume={66},
	number={4},
	pages={893--917},
	year={2018},
	publisher={INFORMS}
}

@inproceedings{wu2019information,
	title={Information design for regulating traffic flows under uncertain network state},
	author={Wu, Manxi and Amin, Saurabh},
	booktitle={2019 57th Annual Allerton Conference on Communication, Control, and Computing (Allerton)},
	pages={671--678},
	year={2019},
	organization={IEEE}
}

@article{massicot2021competitive,
	title={Competitive comparisons of strategic information provision policies in network routing games},
	author={Massicot, Olivier and Langbort, C{\'e}dric},
	journal={IEEE Transactions on Control of Network Systems},
	volume={9},
	number={4},
	pages={1589--1599},
	year={2021},
	publisher={IEEE}
}

@article{mehr2019will,
	title={How will the presence of autonomous vehicles affect the equilibrium state of traffic networks?},
	author={Mehr, Negar and Horowitz, Roberto},
	journal={IEEE Transactions on Control of Network Systems},
	volume={7},
	number={1},
	pages={96--105},
	year={2019},
	publisher={IEEE}
}

@incollection{balakrishna2013information,
	author={Balakrishna, Ramachandran
	and Ben-Akiva, Moshe
	and Bottom, Jon
	and Gao, Song},
	editor={Ukkusuri, Satish V.
	and Ozbay, Kaan},
	title={Information Impacts on Traveler Behavior and Network Performance: State of Knowledge and Future Directions},
	bookTitle={Advances in Dynamic Network Modeling in Complex Transportation Systems},
	year={2013},
	publisher={Springer New York},
	address={New York, NY},
	pages={193--224},
	isbn={978-1-4614-6243-9}
}

@article{zhu2022information,
	title={Information design in nonatomic routing games with partial participation: Computation and properties},
	author={Zhu, Yixian and Savla, Ketan},
	journal={IEEE Transactions on Control of Network Systems},
	volume={9},
	number={2},
	pages={613--624},
	year={2022},
	publisher={IEEE}
}

@article{liu2016effects,
	title={Effects of information heterogeneity in Bayesian routing games},
	author={Liu, Jeffrey and Amin, Saurabh and Schwartz, Galina},
	journal={arXiv preprint arXiv:1603.08853},
	year={2016}
}

@article{ben1991dynamic,
	title={Dynamic network models and driver information systems},
	author={Ben-Akiva, Moshe and De Palma, Andre and Isam, Kaysi},
	journal={Transportation Research Part A: General},
	volume={25},
	number={5},
	pages={251--266},
	year={1991},
	publisher={Elsevier}
}

\end{document}